\documentclass[submission,copyright,creativecommons]{eptcs}
\providecommand{\event}{ACL2 2018} % Name of the event you are submitting to
\usepackage[utf8]{inputenc}
\usepackage{amsmath}
\usepackage[ruled]{algorithm2e}
\usepackage{amssymb}
\usepackage{amsthm}
\usepackage{proof}
\usepackage{underscore}
\usepackage{color,colortbl}
\usepackage{graphicx}
\graphicspath{{images/}}
\usepackage{float}
\usepackage{wrapfig}
\newtheorem{definition}{Definition}[section]

\newtheorem{lemma}[definition]{Lemma}

\newtheorem{example}[definition]{Example}

\newcommand{\newvector}{\vec{\eta}}
\newcommand{\oldvector}{\vec{v}}
\newcommand{\gand}{\cap}
\newcommand{\gnot}{\,^\sim}

\newcommand{\bnot}{\,^\lnot}
\newcommand{\band}{\land}
\newcommand{\bor}{\lor}

\newcommand{\Rn}[1]{\lowercase\expandafter{\romannumeral #1\relax}}
\newcommand{\vto}{\overset{\vec{v}}{\longrightarrow}}
\newcommand{\vtos}{\mathrel{\overset{\vec{v}}{\longrightarrow}\!\!{}^*}}
\newcommand{\vint}{\overset{\vec{v}}{\rightarrowtail}}
\newcommand{\vgen}{\overset{\vec{v}}{\Rightarrow}}

\newcommand{\norm}[1]{\left \langle{#1}\right \rangle}
\newcommand{\normp}[1]{\norm{\norm{#1}}}

\newcommand{\primeII}{\prime\prime}
\newcommand{\primeIII}{\prime\prime\prime}
\newcommand{\primeIV}{\prime\prime\prime\prime}

\newcommand\blfootnote[1]{%
  \begingroup
  \renewcommand\thefootnote{}\footnote{#1}%
  \addtocounter{footnote}{-1}%
  \endgroup
}

\DeclareMathOperator{\ispan}{span}
\DeclareMathOperator{\invmod}{invmod}

\definecolor{Gray}{gray}{0.9}

\def\authorrunning{David Greve \& Andrew Gacek}
\def\titlerunning{Trapezoidal Generalization}

\author{
  David Greve\\
  \texttt{david.greve@rockwellcollins.com}
  \and
  Andrew Gacek\\
  \texttt{andrew.gacek@rockwellcollins.com}
}

\title{Trapezoidal Generalization over Linear Constraints}

\begin{document}

\maketitle
\begin{abstract}

We are developing a model-based fuzzing framework that employs
mathematical models of system behavior to guide the fuzzing
process. Whereas traditional fuzzing frameworks generate tests
randomly, a model-based framework can deduce tests from a behavioral
model using a constraint solver.  Because the state space being
explored by the fuzzer is often large, the rapid generation of test
vectors is crucial.  The need to generate tests quickly, however, is
antithetical to the use of a constraint solver.  Our solution to this
problem is to use the constraint solver to generate an initial
solution, to generalize that solution relative to the system model,
and then to perform rapid, repeated, randomized sampling of the
generalized solution space to generate fuzzing tests.  Crucial to the
success of this endeavor is a generalization procedure with reasonable
size and performance costs that produces generalized solution spaces
that can be sampled efficiently.  This paper describes a
generalization technique for logical formulae expressed in terms of
Boolean combinations of linear constraints that meets the unique
performance requirements of model-based fuzzing.  The technique
represents generalizations using \emph{trapezoidal} solution sets
consisting of ordered, hierarchical conjunctions of linear constraints
that are more expressive than simple intervals but are more efficient
to manipulate and sample than generic polytopes.  Supporting materials
contain an ACL2 proof that verifies the correctness of a low-level
implementation of the generalization algorithm against a 
specification of generalization correctness.  Finally a
post-processing procedure is described that results in a restricted
trapezoidal solution that can be sampled (solved) rapidly and
efficiently without backtracking, even for integer domains.  While
informal correctness arguments are provided, a formal proof of the
correctness of the restriction algorithm remains as future
work.\blfootnote{ This research was developed with funding from the
  Defense Advanced Research Projects Agency (DARPA) under DARPA/AFRL
  Contract FA8750-16-C-0218. The views, opinions and/or findings
  expressed are those of the author(s) and should not be interpreted
  as representing the official views or policies of the Department of
  Defense or the U.S. Government.  Approved for Public Release,
  Distribution Unlimited}

\end{abstract}

\section{Motivation}

Fuzzing is a form of robustness testing in which random, invalid or
unusual inputs are applied while monitoring the overall health of the
system.  Model-based fuzzing is a fuzzing technique that employs a
mathematical model of system behavior to guide the fuzzing process and
explore behaviors that would be difficult to reach by chance.  Whereas
many fuzzing frameworks generate tests randomly, a model-based
framework can deduce tests from a behavioral model using a constraint
solver.  Because the state space being explored by the fuzzer is
generally large, the rapid generation of test vectors is crucial.
Unfortunately, the need to generate tests quickly is antithetical to
the use of a constraint solver.  Our solution to this problem is to
use a constraint solver to generate an initial solution and then to
generalize that solution relative to the constraint.  Test generation
in our model-based fuzzing framework, therefore, consists of repeated,
randomized sampling of the generalized solution space.  Generalization
is crucial to the performance of our framework because it allows us
to decouple constraint solving (slow) from test generation (fast).

\begin{figure}[ht]
\centering
\begin{minipage}{0.65\linewidth}
\begin{algorithm}[H]
  \DontPrintSemicolon
  \While{True}{
    $L \longleftarrow nextConstraint()$\;
    $\vec{v},SAT \longleftarrow solve(L)$\;
    \If{SAT} {
      $T \longleftarrow generalize(L,\vec{v})$\;
      $T^{\prime},\sigma \longleftarrow restrict(T,\vec{v})$\;
      \For{A While} {
        $\vec{w} \longleftarrow sample(T^{\prime},\sigma)$\;
        $fuzzTarget(\vec{w})$\;
      }
    }
  }
  \caption{Conceptual Model-Based Fuzzing Algorithm\label{Fuzzing}}
\end{algorithm}
\end{minipage}
\end{figure}

The fuzzing algorithm, outlined in Algorithm~\ref{Fuzzing}, begins by
identifying an appropriate logical constraint, selected according to
some testing heuristic.  A constraint solver then attempts to generate
a variable assignment that satisfies the constraint.  If it succeeds,
the solution is generalized relative to the constraint.  The
generalization is then restricted so that it can be sampled
efficiently.  The restricted generalization is then repeatedly sampled
to generate random vectors (known to satisfy the constraint) that are
then applied to the fuzzing target.  This process is repeated for the
duration of the fuzzing session.

Trapezoidal generalization is uniquely suited for use in model-based
fuzzing.  The trapezoidal generalization process is reasonably
efficient, both in terms of speed and the size of the final
representation.  Unlike interval generalization, trapezoidal
generalization is capable of representing many linear model features
exactly, allowing sampled tests to better target relevant model
behaviors.  Finally, unlike generic polytope generalization,
trapezoidal generalization supports efficient sampling of the solution
space, enabling rapid test generation and addressing the performance
requirements of model-based fuzzing.

This paper provides details on our trapezoidal generalization,
restriction, and sampling algorithms and the supporting material
includes an ACL2 proof of the correctness of a low-level
implementation of generalization.  Section~\ref{sec:Generalization}
describes the trapezoidal data structure and show how it can be used
to generalize solutions relative to logical formulae expressed as
Boolean combinations of linear constraints.  It includes a formal
specification of generalization correctness and an outline of crucial
aspects of the correctness proof for our generalization algorithm.
Section~\ref{sec:Sampling} describes sampling and provides
details of the restriction process that refines a trapezoid so that it
can be rapidly and efficiently sampled (solved) without backtracking,
even for integer domains.  While informal correctness arguments are
provided, a formal proof of the correctness of the restriction
algorithm remains as future work.  Section~\ref{sec:Formalization}
provides background on our implementation and formalization of
trapezoidal generalization and Section~\ref{sec:Related} compares our
experience using both trapezoidal and interval generalization and
provides an overview of related work in the field.

\subsection{Problem Syntax}

Let $x_1, \ldots, x_n$ be variables of mixed integer and
rational types.  Each variable has an associated numeric
\emph{dimension}, denoted by the variable subscript, that establishes
a complete ordering among the variables.  Numeric expressions are
represented as polynomials (P) over variables having the form $c_nx_n
+ \cdots c_1x_1 + c_0$ where each $c_i$ is a rational constant.  We
specifically restrict our attention to linear, rational, multivariate
polynomials.

\begin{equation*}
    P \equiv c_nx_n + \cdots c_1x_1 + c_0
\end{equation*}

The largest $k$ for which $c_k$ is non-zero is called the dimension of
a polynomial. We sometimes write $P_k$ to explicitly identify the
dimension of a polynomial, $\dim(P_k) = k$.  Note that constants have
dimension 0.

A vector is a sequence of values ordered by dimension.  Given a vector
$\vec{v} = v_n, \ldots, v_1$, we write $E[\vec{v}]$ to denote the
evaluation of an expression $E$ where each $x_i$ in that expression is
replaced by $v_i$.  We assume that every variable has an associated
value in the vector and consider only \emph{consistent vectors},
vectors where each dimension's value is consistent with it's
corresponding variable's type, either rational or integral.  Constants
are self-evaluating and the evaluation operation distributes over the
standard operations in the expected ways.

The atoms in our formulae are linear constraints.  A linear constraint
(L) is an equality or inequality over polynomials.  We will often
write $\prec$ to stand for either $<$ or $\leq$ and $\succ$ to stand
for either $>$ or $\geq$.

\begin{equation*}
    L \equiv P_i = P_j \;|\; P_i > P_j  \;|\; P_i \geq P_j \;|\; P_i < P_j  \;|\; P_i \leq P_j
\end{equation*}

Logical formulae (F) are defined over linear constraints using
conjunction, disjunction, and negation.

\begin{equation*}
    F \equiv F \band F \;|\; F \bor F \;|\; \bnot F \;|\; L
\end{equation*}

We define a variable \emph{bound} (B) as a linear constraint on a single
variable.  The dimension of a variable bound is the dimension of the
bound variable.  We say that a variable bound is \emph{normalized} if the
dimension of the variable is greater than the dimension of the
bounding polynomial, $n > m$.

\begin{equation*}
    B_n \equiv (x_n \prec P_m) \;|\; (x_n \succ P_m) \;|\; (x_n = P_m)
\end{equation*}

We define a \emph{trapezoid} (T) as a (possibly empty) set of variable
bounds.  A trapezoid can be interpreted logically as the conjunction
of all of its variable bounds or geometrically as the intersection of
the volumes defined by the planes bounding each variable.  We say that
a vector \emph{satisfies} a trapezoid if it is a consistent vector and
the logical interpretation of the trapezoid evaluates to True at the
vector.  A trapezoid is \emph{satisfiable} if a consistent vector
exists for which it evaluates to True.  Equivalently, we can say that
a trapezoid \emph{contains} a vector if the point defined by that
vector resides inside of the volumes defined by its geometric
interpretation.  A satisfiable trapezoid contains at least one point.
Below is a grammar for trapezoids expressed in terms of the
intersection of variable bounds.

\begin{equation*}
    T \equiv \ B \gand T \;|\; B \;|\; \varnothing
\end{equation*}

In addition the above syntactic restrictions, we also require that the
constituent bounds of a trapezoid satisfy the following three
properties:

\begin{enumerate}
  \item All variable bounds must be normalized
  \item The set of bounds must be simultaneously satisfiable as witnessed by a consistent reference vector $\vec{v}$
  \item Each $x_n$ is either unbound or is bound by either a single equality or at most one
        upper bound and at most one lower bound.
\end{enumerate}

\begin{wrapfigure}{r}{0.4\textwidth}
  \centering
  \includegraphics[width=0.4\textwidth]{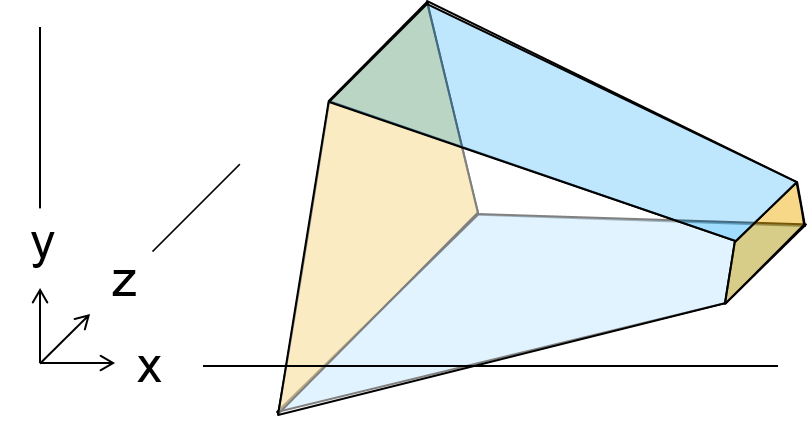}
  \caption{3-d Trapezoidal Volume}
\end{wrapfigure}

We call this data structure a trapezoid because, for every dimension
greater than one, the bounds on the variable of that dimension are, in
general, linear relations expressed in terms of the preceding
variable.  Holding all other variables constant, such linear bounds
describe trapezoidal regions; hence the name.

%% \begin{equation*}
%% \begin{aligned}
%%   C + D*x_{n-1}&  \leq x_n     &\ x_n    &\text{$\;\leq E*x_{n-1} + F$} \\
%%              A&  \leq x_{n-1}  &\ x_{n-1} &\text{$\;\leq B$}        
%% \end{aligned}
%% \end{equation*}

The data structure we use to represent
generalizations is a \emph{region}.  A region is defined as either a
trapezoid or the complement of a trapezoid:

\begin{equation*}
    R \equiv \{ T \} \;|\; \gnot \{ T \}
\end{equation*}

\section{Generalization}
\label{sec:Generalization}

An ACL2 proof that a low-level implementation of the following
generalization procedure satisfies our notion of generalization
correctness is provided in the supporting materials accompanying this
paper.  The following presentation provides an overview of the
generalization process but limits discussion of correctness to
observations concerning key steps in this proof along with references
to the supporting materials.

\subsection{Generalization Correctness}
\label{sec:GeneralizationCorrectness}

Given a formula $F$ and a consistent reference vector $\vec{v} = v_n, \ldots, v_1$, the
objective is to compute a region $R$ which correctly generalizes $\vec{v}$
with respect to $F$. We say that a generalization is correct if it
satisfies the following properties:

\begin{itemize}
    \item\textbf{Invariant 1}$\ \ F[\vec{v}] = R[\vec{v}]$
    \item\textbf{Invariant 2.a} If $F[\vec{v}]$, then $\forall \vec{w}.~ R[\vec{w}] \Rightarrow F[\vec{w}]$
    \item\textbf{Invariant 2.b} If not $F[\vec{v}]$, then $\forall \vec{w}.~ F[\vec{w}] \Rightarrow R[\vec{w}]$
\end{itemize}

The first invariant ensures that the reference vector is contained in
the generalization iff the vector satisfies the original formula.  The second
two invariants ensure that the generalized result is a
\emph{conservative under-approximation} of the original formula.
Together they guarantee that the state space that contains the
reference vector in the generalized result is a subset of the state
space that contains the reference vector in the original formula.

\subsection{Generalization Process}

The generalization process involves transforming a reference vector
and a logical formula into a trapezoidal region.  Generalizing
relative to linear formula (Section~\ref{sec:LinearFormulae}) begins
by transforming linear relations into regions
(Section~\ref{sec:LinearRelations}) and it proceeds by alternately
intersecting and complementing the resulting regions
(Section~\ref{sec:RegionOperations}).  The process of intersecting
regions may further involve generalizing the intersection of sets of
linear bounds into trapezoids (Section~\ref{sec:BoundIntersections}).
The rules for each of these operations is presented in a bottom-up
fashion.

\subsection{Generalizing Linear Relations}
\label{sec:LinearRelations}

We use $\normp{L} \vto R$ to represent the generalization of a linear
relation ($L$) into a trapezoidal region ($R$), a process which may
involve several steps.  First, arbitrary linear relations between two
polynomials can be expressed as linear relations between a polynomial
and zero.

\begin{align*}
    \normp{P_i \prec P_j} \vto \normp{P_i - P_j \prec 0} %%  \label{ator-1}
    \hspace{8pt}
    \normp{P_i \succ P_j} \vto \normp{P_j - P_i \prec 0} %%  \label{ator-2}
    \hspace{8pt}
    \normp{P_i = P_j}     \vto \normp{P_i - P_j = 0}
\end{align*}

Constant linear relations normalize into either the true (empty)
trapezoid or its complement.  

\begin{align*}
    \normp{c_0 \prec 0} &\vto \ \{\}    \hspace{10pt} \mbox{if $c_0 \prec 0$}%%  \label{ator-1}
    \hspace{55pt}
    \normp{c_0 = 0}     \vto \ \{\}     \hspace{10pt} \mbox{if $c_0 = 0$}\\%%  \label{ator-1}
    \normp{c_0 \prec 0} &\vto \gnot\{\} \hspace{10pt} \mbox{if $\bnot (c_0 \prec 0)$}   %%  \label{ator-1}
    \hspace{35pt}
    \normp{c_0 = 0}     \vto \gnot\{\}  \hspace{10pt} \mbox{if $\bnot (c_0 = 0)$}   %%  \label{ator-1}
\end{align*}

Non-constant linear relations between polynomials and zero are
normalized so that they relate the variable with the largest dimension
to a polynomial consisting only of smaller variables and constants.

\begin{align*}
    \normp{P_n <    0} &\vto \normp{x_n < -(P_n/c_n - x_n)}    \hspace{35pt} \mbox{if $c_n > 0$} \\%%  \label{ator-1}
    \normp{P_n <    0} &\vto \normp{-(P_n/c_n - x_n) < x_n}    \hspace{35pt} \mbox{if $c_n < 0$} \\%%  \label{ator-2}
    \normp{P_n \leq 0} &\vto \normp{x_n \leq -(P_n/c_n - x_n)} \hspace{35pt} \mbox{if $c_n > 0$} \\%%  \label{ator-3}
    \normp{P_n \leq 0} &\vto \normp{-(P_n/c_n - x_n) \leq x_n} \hspace{35pt} \mbox{if $c_n < 0$} \\%%  \label{ator-4}
    \normp{P_n = 0}    &\vto \normp{x_n = -(P_n/c_n - x_n)}    \hspace{35pt} \mbox{}               %%  \label{ator-5}
\end{align*}

Normalized linear relations that are true at the reference vector
simply become singleton trapezoidal regions.

\begin{align*}
    \normp{x_n \prec P}   \vto \{(x_n \prec P  )\}  \hspace{35pt} \mbox{if $x_n[\vec{v}] \prec P[\vec{v}]$}  \\%%\label{ator-1}
    \normp{P   \prec x_n} \vto \{(P   \prec x_n)\}  \hspace{35pt} \mbox{if $P[\vec{v}] \prec x_n[\vec{v}]$}  \\%%\label{ator-1}
    \normp{x_n = P}       \vto \{(x_n =     P  )\}  \hspace{35pt} \mbox{if $x_n[\vec{v}] = P[\vec{v}]$}        %%\label{ator-1}
\end{align*}

A normalized linear relation that evaluates to false at the reference
vector, however, is expressed as a negated trapezoidal region
containing a single linear relation that is true at the reference
vector.

\begin{align*}
    \normp{x_n < P   } &\vto \gnot \{(P   \leq x_n)\} \hspace{35pt} \mbox{if $\bnot(x_n[\vec{v}] < P[\vec{v}]   )$}  \\%%\label{ator-1}
    \normp{x_n \leq P} &\vto \gnot \{(P   <    x_n)\} \hspace{35pt} \mbox{if $\bnot(x_n[\vec{v}] \leq P[\vec{v}])$}  \\%%\label{ator-1}
    \normp{P < x_n   } &\vto \gnot \{(x_n \leq P  )\} \hspace{35pt} \mbox{if $\bnot(x_n[\vec{v}] > P[\vec{v}]   )$}  \\%%\label{ator-1}
    \normp{P \leq x_n} &\vto \gnot \{(x_n <    P  )\} \hspace{35pt} \mbox{if $\bnot(x_n[\vec{v}] \geq P[\vec{v}])$}  \\%%\label{ator-1}
    \tag{neq.1}\label{eq:neq1}
    \normp{x_n = P   } &\vto \gnot \{(P   < x_n)\}    \hspace{35pt} \mbox{if $     (x_n[\vec{v}] > P[\vec{v}]   )$}  \\ 
    \tag{neq.2}\label{eq:neq2}
    \normp{x_n = P   } &\vto \gnot \{(x_n < P  )\}    \hspace{35pt} \mbox{if $     (x_n[\vec{v}] < P[\vec{v}]   )$}  %%\label{ator-6}
\end{align*}

Normalizing linear relations in this way ensures that every variable
bound contained in the body of a trapezoid is true at the reference
vector.  It also ensures that regions that contain the reference
vector simplify into simple trapezoids while those that do not are
expressed as the complement of a trapezoid that does.

With the exception of rules \ref{eq:neq1} and \ref{eq:neq2}, the rules
for generalizing linear relations ensure that the generalization
is equal to the original formula (i.e. $ \forall w.~ F[\vec{w}] =
L[\vec{w}]$), trivially satisfying our correctness
invariants. Crucially, rules \ref{eq:neq1} and \ref{eq:neq2} do
satisfy Invariants 1 and 2.b (and, trivially, 2.a) as expressed in
Section~\ref{sec:GeneralizationCorrectness}, ensuring that, in all
cases, our generalization of linear constraints is correct.  See the
functions \texttt{normalize-equal-0} and \texttt{normalize-gt-0} and
their associated lemmas in the file \texttt{top.lisp} of the
supporting materials.

\subsection{Generalizing Bound Intersections}
\label{sec:BoundIntersections}

Let $C$ be a set of normalized bounds known to be satisfiable as
witnessed by the consistent reference vector $\vec{v}$.  The relation $C \vtos T$
represents the fixed-point of intersecting and reducing the various
bounds contained in $C$ into a trapezoidal region.  The following
collection of rules govern the individual intersections of the members
of such a set of bounds to produce a trapezoid.  The rules of the
rewrite system operate on an un-ordered set of bounds, allowing for
arbitrary reordering of the bounds.  The key to this process are rules
for eliminating multiple bounds on the same variable.  The rules for
eliminating multiple upper bounds are shown below. The rules for
multiple lower bounds are symmetric.

\begin{align}
    (x_n < P)    \gand (x_n < Q)    \vto (x_n < P)    \gand \normp{P \leq Q} \hspace{35pt} \mbox{if $P[\vec{v}] \leq Q[\vec{v}]$}  \label{lt-int-1}\tag{lt-int.1} \\
    (x_n < P)    \gand (x_n \leq Q) \vto (x_n < P)    \gand \normp{P \leq Q} \hspace{35pt} \mbox{if $P[\vec{v}] \leq Q[\vec{v}]$}  \label{lt-int-2}\tag{lt-int.2} \\
    (x_n \leq P) \gand (x_n < Q)    \vto (x_n \leq P) \gand \normp{P <    Q} \hspace{35pt} \mbox{if $P[\vec{v}] < Q[\vec{v}]$}     \label{lt-int-3}\tag{lt-int.3} \\
    (x_n \leq P) \gand (x_n \leq Q) \vto (x_n \leq P) \gand \normp{P \leq Q} \hspace{35pt} \mbox{if $P[\vec{v}] \leq Q[\vec{v}]$}  \label{lt-int-4}\tag{lt-int.4}
\end{align}

Note that while $\normp{P \prec Q}$ normalizes the linear relation as
described in Section~\ref{sec:LinearRelations} this bound will never
evaluate to false due to the rules' side-conditions.
Consequently the result will always be a simple trapezoid which is
either empty (if both bounds are constant) or contains a single
normalized bound on a variable whose dimension is less than that of
$x_n$.  Similar observations apply to the rules for simplifying bounds
involving equality, as shown below.  This process and it's correctness
is captured by the function
\texttt{andTrue\--variableBound\--variableBound} and its associated
lemmas (generated by the \texttt{def::trueAnd} macro) in the file
\texttt{poly-proofs.lisp} of the supporting materials.

\begin{align*}
    (x_n = P) \gand (x_n \prec Q) \vto (x_n = P) \gand \normp{P \prec Q} \tag{eq-int.1}\label{eq-int-1}\\
    (x_n = P) \gand (Q \prec x_n) \vto (x_n = P) \gand \normp{Q \prec P} \tag{eq-int.2}\label{eq-int-2}\\
    (x_n = P) \gand (x_n = Q)     \vto (x_n = P) \gand \normp{P = Q}     \tag{eq-int.3}\label{eq-int-3}
\end{align*}

\begin{lemma}
The intersection rules are terminating.
\end{lemma}
\begin{proof}
Define a measure on the conjunction of bounds to be the sum of the
dimension of each bound. Each rule of the intersection rewrite system
reduces this measure. For example, consider rule \ref{lt-int-1}. We
know $dim(x) > dim(P)$ and $dim(x) > dim(Q)$, and thus $dim(x) >
dim(\langle P \leq Q\rangle)$ and in particular $dim(x < Q) >
dim(\langle P \leq Q\rangle)$. Since the measure is always
non-negative, the intersection rules are terminating.  See the
\texttt{def::total} event for the function \texttt{intersect} in the
file \texttt{intersection.lisp} of the supporting materials for a proof
of termination for ordered (not un-ordered, see Section~\ref{sec:OrderedTrapezoids}) bounds.
\end{proof}

\begin{lemma}
When run to completion, the intersection rules result in a trapezoid.
\end{lemma}
\begin{proof}
The intersection rules apply to a set of bounds that are normalized
and satisfiable.  To qualify as a trapezoid the only missing
requirement is that there may be multiple upper, lower, and equality
bounds on each variables. For multiple upper bounds, note that the 4
rules above cover all cases. Although rule \ref{lt-int-3} does not
cover the case when $P[\vec{v}] = Q[\vec{v}]$, that case will in fact
be covered by rule \ref{lt-int-2} since then $Q[\vec{v}] \leq
P[\vec{v}]$. The case of multiple lower bounds is
symmetric. Therefore, the rewrite rules will eliminate all multiple
and lower and upper bounds so that the result is a trapezoid.  See
lemma \texttt{trapezoid-p-intersect} in the file
\texttt{intersection.lisp} of the supporting materials.
\end{proof}

While, in general, the the fixed-point resulting from intersecting and
reducing a set of bounds will differ from the original set, each rule
that we apply to perform this normalization satisfies Invariants 1 and
2.a (and trivially 2.b) as expressed in
Section~\ref{sec:GeneralizationCorrectness}, ensuring the correctness
of this process.  

\subsection{Generalizing Region Intersections and Complements}
\label{sec:RegionOperations}

The relation $R^a \cap R^b \vint R^c$ indicates that region $R^c$ is a
generalized intersection of regions $R^a$ and $R^b$ relative to $\vec{v}$
and $\gnot R^a \vint R^b$ to indicate that $R^b$ is the generalized
complement of $R^a$.  We characterize the behavior of region
intersection and complement with the following set of rules:

\begin{align*}
    \infer[\mbox{\textsc{comp}}]
          {\gnot{\gnot{\{T\}}} \vint \{T\}}
          {}
    \hspace{35pt}
    \infer[\mbox{\textsc{tint}}]
          {\{T^a\} \cap \{T^b\} \vint \{T^c\}}
          {T^a \cap T^b \vtos T^c}\\
    \\
    \infer[\mbox{\textsc{cint.1}}]
          {\gnot{\{T\}} \cap R \vint \gnot{\{T\}}}
          {}
    \hspace{35pt}
    \infer[\mbox{\textsc{cint.2}}]
          {R \cap \gnot{\{T\}} \vint \gnot{\{T\}}}
          {}
\end{align*}

While perhaps surprising, the rules \mbox{\textsc{cint.1}} and
\mbox{\textsc{cint.2}} not only simplify the region intersection
process, they are also both consistent with and essential for
preserving generalization correctness (specifically, Invariant 2.b) as
described in Section~\ref{sec:GeneralizationCorrectness}.  See the
functions \texttt{and-regions} and \texttt{not-region} and their
associated lemmas in the file \texttt{top.lisp} from the supporting
materials.

\subsection{Generalizing Linear Formulae}
\label{sec:LinearFormulae}

We write $F \vgen R$ to mean that region $R$ generalizes $\vec{v}$
with respect to formula $F$. In general, conjunction generalizes into
region intersection, negation generalizes into region complement, and
disjunction generalizes into a complement of the intersection of the
complements of the generalized arguments.  The generalization relation
is defined inductively over the structure of the formula $F$.

\begin{align*}
    \infer[\mbox{\textsc{gen-atom}}]
          {L \vgen R}
          {\normp{L} \vto R}
    \hspace{35pt}
    \infer[\mbox{\textsc{gen-and}}]
          {F_1 \band F_2 \vgen R_1 \gand R_2}
          {F_1 \vgen R_1 & F_2 \vgen R_2}
\end{align*}
\begin{align*}
    \infer[\mbox{\textsc{gen-not}}]
          {\bnot F \vgen \gnot R}
          {F \vgen R}
    \hspace{35pt}
    \infer[\mbox{\textsc{gen-or}}]
          {F_1 \bor F_2 \vgen \gnot(\gnot R_1 \gand \gnot R_2)}
          {F_1 \vgen R_1 & F_2 \vgen R_2}
\end{align*}

This procedure can be shown to satisfy our correctness invariants by
induction over the structure of a linear formula and by appealing to
the correctness of the generalization of linear relations and region
intersection.  See the definition \texttt{generalize-ineq} and the
theorems \texttt{inv1-generalize-ineq} and
\texttt{inv2-generalize-ineq} in the file \texttt{top.lisp} of the
supporting materials.

\section{Sampling}
\label{sec:Sampling}

\emph{Sampling} is the process of identifying randomized, concrete,
type-consistent variable assignments (vectors) that satisfy all of the variable bounds in a
trapezoid \footnote{We focus on trapezoids because solving a trapezoid
  complement is trivial.  The complement of a trapezoid is merely a
  disjunction of negated linear bounds.  An assignment that satisfies
  any one of the negated bounds satisfies the trapezoid
  complement.}. We call this process sampling, rather than solving, to
emphasize the fact that it is intended to be randomized, simple, and
efficient.  Indeed, the structure of a trapezoid lends itself to a
simple and efficient sampling algorithm.  The trapezoidal structure
ensures that the variable with the smallest dimension is bounded only
by constants\footnote{Technically, variables may also be unbound.  In
  such case arbitrary value assignments suffice.}.  Sampling of the
smallest dimension variable therefore involves simply choosing a value
consistent with its type and constant bounds.  Because each normalized variable
bound is expressed strictly in terms of variables of smaller
dimension, sampling the variables by proceeding from the smallest to
the largest dimension guarantees that the bounds on each subsequent
variable will evaluate to constants relative to variables already
assigned.  Sampling each variable from the smallest to the largest dimension,
therefore, involves nothing more than choosing a value for each
variable consistent with type and constant bounds computed based on the
previous variable assignments.

While the above algorithm is certainly simple and efficient, the
trapezoidal data structure as defined is not sufficient to guarantee
that every variable assignment that satisfies the first $N$ bounds
will satisfy the $N+1^{th}$ bound.  A given set of variable
assignments may result in an upper bound that is smaller than a
lower bound or, more subtly, an integer variable may be constrained by
bounds that do not permit an integral solution.  While backtracking
and attempting different sets of variable assignments is possible in
such cases, it can also be very inefficient.

To address this concern we present an overview of a post-processing
step for trapezoids, performed before sampling, that results in a
restricted trapezoid with properties that ensure that the simple and
efficient sampling algorithm presented above will never encounter a
bound that cannot be satisfied (even for integer variables) and thus
never needs to backtrack.  While the following procedure has not been
mechanically verified, we sketch selected aspects of the informal
proof with the long term objective of formalizing and verifying its
correctness in ACL2.

\subsubsection{Ordered Trapezoids}
\label{sec:OrderedTrapezoids}

To better explain our post-processing approach we first introduce a
more refined representation for trapezoids, beginning with the concept
of variable \emph{intervals}.  A variable interval is either a single
variable bound or a pair of bounds consisting of an upper and lower
bound on the same variable.  The dimension of an interval is simply
the dimension of the bound variable.  Variable intervals allow us to
gather all of the linear bounds relevant to each variable in one
place.

\begin{equation*}
    I_n \equiv B_n \;|\; (P^L \prec x_n) \gand (x_n \prec P^U)
\end{equation*}

Variable intervals can be used to construct ordered trapezoids.  An
ordered trapezoid is simply a sequence of variable intervals organized
in descending order based on dimension, terminating in the empty
trapezoid.  Note that these new syntactic constructs do not change in
any way the nature of trapezoids other than to impose additional ordering
constraints on their representation.

\begin{align*}
    T_n &\equiv \ I_n \gand T_m \ \ \mbox{where $n > m$} \\
    T_0 &\equiv \varnothing
\end{align*}

Using ordered trapezoids, we can define the process for sampling a
trapezoid, $\epsilon(T)$ to produce a consistent vector that satisfies
the trapezoid.  Note that $\epsilon(I_n[\vec{w}])$ represents a random
choice of a value consistent in type with variable $n$ and satisfying
the interval $I_n$ evaluated at the vector $\vec{w}$ (assuming such a
value exists), $\{\}$ represents an empty vector, and $\vec{w}[i] = x$
represents an update of vector $\vec{w}$ such that dimension $i$ is
equal to the value $x$.

\begin{align*}
    \infer[]
          {\epsilon(\varnothing) \rightarrow \{\}}
          {}
    \hspace{35pt}
    \infer[]
          {\epsilon(I_n \gand T_m) \rightarrow \vec{w}[n] = \epsilon(I_n[\vec{w}])}
          {\epsilon(T_m) \rightarrow \vec{w}}
\end{align*}

\subsection{Restriction}

The purpose of post-processing is to restrict the domain of values that
can be assigned to selected variables in the trapezoid to ensure that
the sampling process can proceed without backtracking.  We call the
trapezoidal post-processing step \emph{restriction}.  The restriction
process takes as input a trapezoid, a change of basis ($\sigma$, initially
an identity function), and a consistent reference vector.  The output of the
restriction process is a new (restricted) trapezoid and a change of
basis that maps vectors sampled from the restricted trapezoid back
into the vector space of the original trapezoid.

\begin{align*}
    T,\,\sigma,\,\vec{v} \overset{\textit{R}}\longrightarrow T^{\prime},\,\sigma^{\prime}
\end{align*}

The restriction process is specified as a single pass over the
trapezoid intervals, starting with the largest interval and proceeding
to the smallest.  Each interval in the trapezoid is restricted by
applying the steps summarized below.

\begin{itemize}
  \item\textbf{Bound Fixing (BF)} Ensures that linear bounds are expressed
    only in terms of $\geq$ and $\leq$.
  \item\textbf{Integer Equality (IE)} Ensures that
    equalities involving integer variables are always satisfiable.
  \item\textbf{Interval Restriction (IR)} Ensures that
    intervals are always satisfiable.
\end{itemize}

Each interval in the trapezoid may suggest a domain
restriction, in the form of a set of linear constraints, and a change
of basis in the form of a linear transformation.  Domain restrictions
are intersected with the remaining trapezoid before it, too, is
restricted.  Changes of base are accumulated and applied to the
current interval, the remaining trapezoid, and the reference vector.
When complete, the final restricted trapezoid and the final change of
basis are returned.  

\begin{align*}
    \infer[\mbox{\sc restriction-base}]
          {\varnothing,\,\sigma,\,\vec{v} \overset{\textit{R}}\longrightarrow \varnothing,\,\sigma}
          {}
\end{align*}

\begin{align*}
    \infer[\mbox{\sc restriction-step}]
          {I_n \gand T_m,\,\sigma,\,\vec{v} \overset{\textit{R}}\longrightarrow I^{\primeIII}_n \gand T^{\primeIV}_m,\,\sigma^{\primeIII}}
          {\begin{aligned}
              I_n \gand T_m,\,\sigma,\,\vec{v}                                      &\overset{\textit{BG}}\longrightarrow I^{\prime}_n \gand T^{\prime}_m,\,\sigma^{\prime},\,{\vec{v}\,}^{\prime}\\
              I^{\prime}_n \gand T^{\prime}_m,\,\sigma^{\prime},\,{\vec{v}\,}^{\prime}      &\overset{\textit{IE}}\longrightarrow I^{\primeII}_n \gand T^{\primeII}_m,\,\sigma^{\primeII},\,{\vec{v}\,}^{\primeII}\\
              I^{\primeII}_n \gand T^{\primeII}_m,\,\sigma^{\primeII},\,{\vec{v}\,}^{\primeII}    &\overset{\textit{IR}}\longrightarrow I^{\primeIII}_n \gand T^{\primeIII}_m,\,\sigma^{\primeIII},\,{\vec{v}\,}^{\primeIII}\\
                                 T^{\primeIII}_m,\,\sigma^{\primeIII},\,{\vec{v}\,}^{\primeIII} &\overset{\textit{R}}\longrightarrow T^{\primeIV}_m,\,\sigma^{\primeIV}
            \end{aligned}
           }
\end{align*}

The restricted trapezoid can then be sampled using our simple sampling
algorithm and the resulting vectors can be mapped back into the
original vector space using the change of basis in such a way that
$\sigma^{\prime}[\epsilon(T^{\prime})]$ is a type consistent vector
and $T[\sigma^{\prime}[\epsilon(T^{\prime})]]$ is true.

The following discussion assumes that all integer valued variables are
bounded only by other integer valued variables.  This can be ensured
if, say, the dimension of each rational variable is always greater
than the dimension of every integer variable.

\subsection{Bound Fixing (BF)}

Bound fixing transforms exclusive inequalities on integer variables
into inclusive inequalities and tightens any remaining integer bounds.

\begin{align*}
    \infer[\mbox{\sc bound-fixing}]
          {I_n \gand T_m,\,\sigma,\,\vec{v} \overset{\textit{BF}}\longrightarrow I^{\prime}_n \gand T_m,\,\sigma,\,\vec{v}}
          {BF(I_n) \rightarrow I^{\prime}_n}
\end{align*}

We omit a more detailed description of this operation for space considerations.

\subsection{Integer Equality (IE)}

Any integer variable equality bound
appearing in a trapezoid can be expressed as:

\begin{align*}
  x_n = (N_k*x_k + \ldots + N_0)/D_n
\end{align*}

where all $x_i$ and $N_i$ are integer valued and $D_n$ is a positive
integer, representing the least common denominator of the bounding
polynomial coefficients.  Given $a, b \in \mathbb{Z}$ we write $(a | b)$, called
``$a$ divides $b$'', if and only if there exists $c \in \mathbb{Z}$
such that $b = a*c$. A divisibility constraint is an expression of the
form $(m | E)$ where $m \in \mathbb{Z}$ and $E \in \mathcal{E}$. A
vector $\vec{v}$ is called a {\em solution} to a divisibility constraint $(m
| E)$ if and only if $(m | E[\vec{v}])$. A divisibility constraint is called
solvable if and only if it has at least one solution.  To ensure that
$x_n$ has an integer solution, $\vec{w}$ must be chosen so that $(D_n |
(N_k*x_k + \ldots + N_0)[\vec{w}]) $.  Of course if $D_n$ is 1, this is
trivial and no additional work is required.

Given a non-trivial, solvable divisibility constraint $(m | E)$ we want
to find an efficient way to enumerate all solutions.  Towards this
end, we introduce the concept of a {\em Trapezoidal Change of Basis} (TCOB):
a variable substitution that preserves dimension.

\begin{definition}[Trapezoidal Change of Basis]
A map $\sigma : \mathcal{E} \to \mathcal{E}$ is called an trapezoidal change of basis, if
\begin{enumerate}
\item $\sigma(c_0 + c_1x_1 + \cdots + c_nx_n) = c_0 + c_1\sigma(x_1) + \cdots + c_n\sigma(x_n)$, i.e. $\sigma$ is a homomorphism, and
\item $\dim(\sigma(x_i)) = \dim(x_i)$, for all $i = 1, \ldots, n$.
\end{enumerate}
\end{definition}

For a solvable divisibility constraint $(m | E)$, we propose to
construct a TCOB $\sigma$ such that 1) every vector is a solution to
the divisibility constraint $(m | \sigma(E))$ and 2) each solution to $(m
| E)$ is represented by a solution to $(m | \sigma(E))$.

%% \begin{example}\label{ex:tcob}
%% Consider the divisibility constraint $(2 | (1 + x + y))$ where $dim(x) = 1$ and $dim(y) = 2$.

%% Let $\sigma_1$ be a TCOB defined by $\sigma_1(x) = 2x + 1$ and
%% $\sigma_1(y) = 2y$. Letting $E = 1 + x + y$ we have $\sigma_1(E) = 2 +
%% 2x + 2y$ so that every vector is a solution to $(2 |
%% \sigma_1(E))$. However, this TCOB constrains $\sigma_1(x)$ to always be
%% odd and $\sigma_1(y)$ to always be even and thus it does not preserve
%% all solutions to $(2 | E)$ such as $(2, 1)$.

%% Let $\sigma_2$ be a TCOB defined by $\sigma_2(x) = x + 1$ and
%% $\sigma_2(y) = 2y + x$. Letting $E = 1 + x + y$ we have $\sigma_2(E) = 2
%% + 2x + 2y$ so that every vector is a solution to $(2 |
%% \sigma_2(E))$. Moreover, every solution to $(2 | E)$ is preserved by this
%% TCOB.
%% \end{example}

We define the notion of {\em span} to clarify when exactly a TCOB preserves a solution:

%% \old = w
%% \new = v

\begin{definition}
Let $\sigma$ be a trapezoidal change of basis defined by $\sigma(x_i)
= E_i$ for each $i = 1, \ldots, n$. Let $\oldvector$ and $\newvector$ be vectors. We say
$\oldvector = \sigma[\newvector]$ if and only if $\oldvector = (E_1[\newvector], \ldots, E_n[\newvector])$. We
define $\ispan(\sigma) = \{ \oldvector : \exists \newvector.~ \oldvector = \sigma[\newvector]\}$.
\end{definition}

%% \begin{example}\label{ex:span}
%% Let $\sigma_1$ and $\sigma_2$ be as in Example~\ref{ex:tcob}. Then we
%% have $\ispan(\sigma_1) = \{(2x + 1, 2y) : x, y \in \mathbb{Z}\}$ and
%% $\ispan(\sigma_2) = \{(x+1, 2y+x) : x, y \in \mathbb{Z} \}$. Note that
%% $(2, 1) \notin \ispan(\sigma_1)$ and $(2, 1) \in \ispan(\sigma_2)$.
%% \end{example}

\begin{lemma}
For every $E \in \mathcal{E}$ we have $E[\sigma[\newvector]] = \sigma(E)[\newvector]$.
\end{lemma}
\begin{proof}
By induction on the structure of $E$.
\end{proof}

\begin{lemma}
\label{tcob}
Let $(m | E)$ be a solvable divisibility constraint. Then there exists a
trapezoidal change of basis $\sigma$ such that
\begin{enumerate}
    \item $\sigma(E) = m E'$ for some $E'$, and
    \item $\forall \oldvector.~ (m | E[\oldvector]) \Longrightarrow \oldvector \in \ispan(\sigma)$.
\end{enumerate}
Note that first condition implies that every vector is a solution to $(m | \sigma(E))$.
\end{lemma}
\begin{proof}
Induction on $n$, the number of dimensions. In the base case $n = 0$
for which case $E = c$ where $c \in \mathbb{Z}$. Take $\sigma$ to be
the identity map on $\mathcal{E}$. Since $(m | c)$ is solvable we have
$c = m d$ for some $d \in \mathbb{Z}$. Thus $\sigma(E) = m d$,
satisfying property 1. The second property is trivially true since
there is only a single (empty) vector of dimension 0.

Now consider the case for $n > 0$ dimensions. Let $E = F + cx$ where
$dim(F) < n$, $c \in \mathbb{Z}$, and $dim(x) = n$. Let $g = \gcd(c,
m)$ and let $c'$ and $m'$ be such that $c = gc'$ and $m = gm'$. Since
$(m | E)$ is solvable, $(g | E)$ must also be solvable. Furthermore, we
know $(g | c)$ so $(g | F)$ is solvable. By the inductive hypothesis, let
$\sigma'$ be a TCOB with
\begin{enumerate}
    \item $\sigma'(F) = g F'$ for some $F'$, and
    \item $\forall \oldvector.~ (g | F[\oldvector]) \Longrightarrow \oldvector \in \ispan(\sigma')$.
\end{enumerate}
Then, let $\sigma$ be an extension of $\sigma'$ with $\sigma(x) = m'x
- \invmod(c', m')F'$. Here $\invmod(c', m')$ is the multiplicative
inverse of $c'$ modulo $m'$ which is well-defined since $\gcd(c', m')
= 1$. Note also that $dim(\sigma(x)) = dim(x)$. Using $\sigma$ we
have,
\begin{align*}
    \sigma(E) &= g F' + c\sigma(x) \\
    &= g F' + c(m'x - \invmod(c', m')F') \\
    &= cm'x + g F' - c\invmod(c', m')F' \\
    &= cm'x + g F' - gc'\invmod(c', m')F' \\
    &= cm'x + g (1 - c'\invmod(c', m'))F' \\
    &= cm'x + g(1 - (1 + m'k)) F' &\mbox{for some $k\in\mathbb{Z}$}\\
    &= cm'x + gm'k F' \\
    &= gc'm'x + gm'k F' \\
    &= gm'(c'x + k F') \\
    &= m(c'x + k F')
\end{align*}
So that $\sigma$ satisfies property 1.

To prove property 2 on $\sigma$, let $\oldvector$ be a vector such
that $(m | E[\oldvector])$. Then we have $(m | (F[\oldvector] +
c\oldvector_x))$ where $\oldvector_x$ is the $x$-component of
$\oldvector$. This implies $(g | F[\oldvector])$ and so by the
inductive hypothesis we have $\oldvector' \in span(\sigma')$ where
$\oldvector'$ is $\oldvector$ without the final $x$-component. This
means there is a vector $\newvector'$ such that 
$\oldvector' = \sigma[\newvector'].$ To show $\oldvector \in \ispan(\sigma)$
we want to extend $\newvector'$ with an $x$-component, $\newvector_x$,
such that $\oldvector = \sigma[\newvector]$.  Such a $\newvector_x$
exist if and only if we can satisfy $\oldvector_x = m'\newvector_x -
\invmod(c', m')F'[\newvector']$. This equation will have an integral
solution for $\newvector_x$ if and only if
\begin{align*}
    \oldvector_x &\equiv -\invmod(c', m')F'[\newvector'] \pmod{m'} &\Longleftrightarrow \\
    \oldvector_x + \invmod(c', m')F'[\newvector'] &\equiv 0 \pmod{m'} &\Longleftrightarrow \\
    c'\oldvector_x + c'\invmod(c', m')F'[\newvector'] &\equiv 0 \pmod{m'} &\Longleftrightarrow \\
    c'\oldvector_x + F'[\newvector'] &\equiv 0 \pmod{m'} &\Longleftrightarrow \\
    gc'\oldvector_x + gF'[\newvector'] &\equiv 0 \pmod{gm'} &\Longleftrightarrow \\
    c\oldvector_x + \sigma(F)[\newvector'] &\equiv 0 \pmod{m} &\Longleftrightarrow \\
    c\oldvector_x + F[\oldvector'] &\equiv 0 \pmod{m}
\end{align*}
Note that $F[\oldvector'] = F[\oldvector]$ since $dim(F) < n$. Thus we only need to show
$(m | (c\oldvector_x + F[\oldvector]))$ which is equivalent to $(m | E[\oldvector])$. This is true by
assumption. Thus property 2 also holds on $\sigma$.
\end{proof}

The proof of Lemma~\ref{tcob} contains an algorithm for
recursively computing a TCOB $\sigma$ for any solvable divisibility
constraint $(m | E)$. We demonstrate this algorithm in the following
example.

\begin{example}
Consider the divisibility constraint $(2 | (1 + x + y))$ where $dim(x) =
1$ and $dim(y) = 2$. The algorithm implicit in the proof of
Lemma~\ref{tcob} produces the TCOB $\sigma(x) = x$, $\sigma(y) =
2y - x - 1$. Letting $E = 1 + x + y$ we have $\sigma(E) = 2y$ so that
$(2 | \sigma(E))$ is true for all vectors.
\end{example}

Given a constraint $x_n = (N_k*x_k + \ldots + N_0)/D_n$ we have shown
that we can compute a TCOB $\sigma$ so that $(D_n | (N_k*x_k + \ldots
+ N_0))$ is always satisfied, allowing us to always compute an
integral value for $x_n$.  We can translate any such satisfying solution ($\newvector$)
into a solution of our original trapezoid via $\oldvector =
\sigma[\newvector]$.  It is easy to show that $\sigma$ preserves all
constraints on trapezoids. Therefore, we can repeatably apply this
procedure starting from the highest dimensions and working our way
down, eliminating all non-trivial integral constraints along the way.
Our final step is to update our reference vector to reflect the change
of base, ${\vec{v}\,}^{\prime} = \sigma^{-1}[\vec{v}] = \{ \vec{w} : \vec{v} = \sigma[\vec{w}]\}$.  We
know that such a vector exists because the span of $\sigma$ includes
all of the vectors that satisfied the original constraint, including
the original reference vector.

\begin{align*}
    \infer[\mbox{\sc integer-equality}]
          {I_n \gand T_m,\,\sigma,\,\vec{v} \overset{\textit{IE}}\longrightarrow I^{\prime}_n \gand T^{\prime}_m,\,\sigma^{\prime},\,{\vec{v}\,}^{\prime}}
          {\begin{aligned}
              &I_n \equiv (x_n = P/D_n)\\
              &x_n \in \mathbb{Z} \quad
              (D_n | P) \rightarrow \sigma_n \quad
              I^{\prime}_n = (x_n = \sigma_n(P)/D_n) \\
              &\sigma^{\prime} = \sigma_n(\sigma) \quad
              {\vec{v}\,}^{\prime} = \sigma^{-1}_n(\vec{v}) \quad
              T^{\prime}_m = \sigma_n(T_m)
            \end{aligned}
           }
\end{align*}

\subsection{Interval Restriction}

While double bounded intervals are known to be satisfiable at the
reference vector, there is no guarantee that they are satisfiable
under arbitrary variable assignments.  To ensure satisfiability, we
need to ensure that, for any valid variable assignment, the upper
bound will not be less than the lower bound and, for integer
variables, that an integer solution exists between the upper and lower
bound.

For all rational variables and for integer variables when $lcd(P^L) =
1$ or $lcd(P^U) = 1$ we simply intersect the domain restriction
$\normp{P^L \leq P^U}$ with the remaining trapezoid.  This ensures
that, for any variable assignment that satisfies the newly restricted
trapezoid, the lower bound will not be greater than the upper bound.
For integer variables this restriction is sufficient because either
$P^L$ or $P^U$ will be an integer valued polynomial and, as a result,
when $P^L \leq P^U$ is true it always contains an integer solution.

\begin{align*}
    \infer[\mbox{\textsc{interval-restriction-1}}]
          {I_n \gand T_m,\,\sigma,\,\vec{v} \overset{\textit{IR}}\longrightarrow I_n \gand T^{\prime}_m,\,\sigma,\,\vec{v}}
          {\begin{aligned}
              &I_n = (P^L \leq x_n) \gand (x_n \leq P^U)\\
              &(\; x_n \in \mathbb{Q} \; \bor \; lcd(P^L) = 1 \; \bor \; lcd(P^U) = 1 \;) \\
              &\normp{{P^L} \leq {P^U}} \gand T_m \vtos T^{\prime}_m
            \end{aligned}
           }
\end{align*}

When $lcd(P^L) > 1$ and $lcd(P^U) > 1$ the constraint $P^L \leq P^U$
is not sufficient to ensure an integer solution for $x$ because $P^L$
and $P^U$ are not guaranteed to be integer valued for every variable
assignment.  The constraint $P^L + 1 \leq P^U$, on the other hand, does
always contain an integer solution.  If this constraint contains the
reference vector (ie: $P^L[\vec{v}] + 1 \leq P^U[\vec{v}]$) then it is conservative
to use this constraint as a domain restriction.

\begin{align*}
    \infer[\mbox{\textsc{interval-restriction-2}}]
          {I_n \gand T_m,\,\sigma,\,\vec{v} \overset{\textit{IR}}\longrightarrow I_n \gand T^{\prime}_m,\,\sigma,\,\vec{v}}
          {\begin{aligned}
              &I_n = (P^L \leq x_n) \gand (x_n \leq P^U)\\
              &x_n \in \mathbb{Z} \\
              &lcd(P^L) > 1 \\
              &lcd(P^U) > 1 \\
              &P^{L}[\vec{v}] + 1 \leq P^{U}[\vec{v}] \\
              &\normp{P^L + 1 \leq P^U} \gand T_m \vtos T^{\prime}_m
            \end{aligned}
           }
\end{align*}

If $lcd(P^L) > 1$, $lcd(P^U) > 1$, and $P^L[\vec{v}] + 1 > P^U[\vec{v}]$ we
perform a change of base to ensure that either the upper or lower
bound is always integral.  Here we consider the change of base for the
lower bound with the understanding that the case of the upper bound is
symmetrical.  We first adjust (tighten) the lower bound in such a way
that it is equal to $x[\vec{v}]$ when evaluated at the reference vector.
Note that, because $P^L[\vec{v}] + 1 > P^U[\vec{v}]$, there is only 1 integer
solution in this interval and $x_n[\vec{v}]$ is it.  We compute a change
of base for $P^{\prime{L}}$ to ensure that it is always integral.  We then
restrict the domain of $T_m$ to ensure that the resulting upper bound
is never less than our new lower bound.

\begin{align*}
    \infer[\mbox{\textsc{interval-restriction-3}}]
          {I_n \gand T_m,\,\sigma,\,\vec{v} \overset{\textit{IR}}\longrightarrow I^{\prime}_n \gand T^{\primeII}_m,\,\sigma^{\prime},\,{\vec{v}\,}^{\prime}}
          {\begin{aligned}
              &I_n = (P^L \leq x_n) \gand (x_n \leq P^U)\\
              &x_n \in \mathbb{Z} \quad lcd(P^L) > 1 \quad lcd(P^U) > 1\\
              &P^{\prime{L}} = P^L + (x_n[\vec{v}] - P^L[\vec{v}])\\
              &D^L = lcd(P^{\prime{L}})\\
              &(D^L | P^{\prime{L}}) \rightarrow \sigma_n\\
              &P^{\prime{U}} = \sigma_n(P^{U})  \quad 
              P^{\primeII{L}} = \sigma_n(P^{\prime{L}})/D^L \\
              &I^{\prime}_n = (P^{\primeII{L}} \leq x_n) \; \gand \; (x_n \leq P^{\prime{U}})\\
              &\sigma^{\prime} = \sigma_n(\sigma) \quad
              {\vec{v}\,}^{\prime} = \sigma^{-1}_n(\vec{v}) \quad
              T^{\prime}_m = \sigma_n(T_m)\\
              &\normp{P^{\primeII{L}} \leq P^{\prime{U}}} \gand T^{\prime}_m {\mathrel{\overset{v^{\prime}}{\longrightarrow}\!\!{}^*}} T^{\primeII}_m
            \end{aligned}
            }
\end{align*}

\section{Implementation and ACL2 Formalization}
\label{sec:Formalization}

FuzzM\cite{FuzzM} is a model-based fuzzing framework, implemented
primarily in Java, that employs Lustre\cite{lustre:ieee} as a modeling
and specification language and leverages counterexamples produced by
the the JKind\cite{JKind} model checker to drive the fuzzing process.
The FuzzM infrastructure includes support for trapezoidal
generalization of JKind counterexamples relative to Lustre models
involving linear constraints, integer division and remainder, and
uninterpreted functions.  Our implementation of trapezoidal
generalization requires approximately 5K lines of Java code.

Early versions of the fuzzing framework leveraged an interval
generalization capability provided by JKind.  At the time we observed
that, while the results of interval generalization are trivial to sample, it 
does not generalize linear model features well.
Trapezoidal generalization was envisioned as an extension of interval
generalization that would preserve the property of efficient sampling
while improving support for linear constraints.  In architecting the
generalization procedure, however, it soon became clear that we didn't
know what it meant for our generalizer to be correct.  To address this
issue we developed a formalization of generalization correctness in
ACL2.  In a rump session of the 2017 ACL2 workshop we presented our
formalization of generalization correctness, a high-level
specification for our implementation, and a proof that our
implementation was correct\cite{rump}.  We also reported that, in proving that
our implementation was correct, we had identified an error in our
initial high-level specification.  Our inability to correctly specify
even the high-level behavior of our system without mechanical
assistance further motivated us to formalize and verify the
correctness of our low-level implementation as well.  Much of our
low-level ACL2 specification has been transliterated from the Java
implementation to ensure that it captures our essential design
decisions.  While the proofs of low-level correctness were useful in
ensuring the absence of off-by-one errors and the like, perhaps the
primary benefit of this activity was in the formalization process
itself.  Just as formalizing our high-level algorithms helped to
clarify our software architecture, formalizing the low-level behavior
of our algorithms has helped us both to clarify and to simplify the
resulting implementation.

We now have a proof that the low-level implementation of our
trapezoidal generalization procedure over linear constraints satisfies
our notion of generalization correctness.  At the heart of our
formalization is a library for manipulating, reasoning about, and
evaluating linear relations over rational polynomials.  On top of this
library we introduce a notion of normalized variable bounds upon which
we define the trapezoidal data structure and a predicate that
recognizes ordered trapezoids and regions.  The proof that our
low-level implementation satisfies our high level specification was
largely an exercise in proving that each low-level function preserved
our set of correctness invariants.  While the majority of this proof
effort was routine, two specific challenges were addressed with
specialized ACL2 libraries.  

First, to prove that functions operating on normalized variable bounds
and trapezoidal regions preserve our variable ordering invariants we
employ an arcane feature of the \emph{nary}\cite{nary} library that,
effectively, enables us to treat subset relations as equivalence
relations in appropriate contexts, rewriting subset terms into their
super-sets.  This feature allows us to specify function contracts using
rewrite rules, rather than back-chaining rules, to establish ordering
relations among the variables appearing in the function's results.  In
the supporting materials see \texttt{set\--upper\--bound\--equiv} and
\texttt{set\--upper\--bound\--ctx} defined by the \texttt{defequiv+}
event in the file \texttt{poly.lisp} along with the congruence rule
\texttt{set\--upper\--bound\--append} and the family of driver rules
exemplified by \texttt{>-all\--upper\--bound\--backchaining}.  See
also the nary equivalence relation
\texttt{set\--upper\--bound\--equiv} as used in
\texttt{set\--upper\--bound\--equiv\--all\--bound\--list\--variables\--restrict}
to specify a subset rewrite for specific variables returned by the
function \texttt{intersect} in the file \texttt{intersect.lisp} and
compare it with the more traditional \texttt{subset-p} rule found just
above in the lemma \texttt{trapezoid-p\--restrict}.

Second, because the termination of the function for performing
top-level trapezoidal intersection is a bit tricky, in that it is
doubly recursive, reflexive, and it relies on the variable ordering property of
trapezoids, we introduced its definition using the
\emph{def::ung}\cite{defung,assuming} macro.  This allowed us to
postpone its termination proof until we had established that the
intersection function preserved the variable ordering.  We then used
\emph{def::total} to prove that the function does, in fact, terminate on
well-formed trapezoids.  These events are found in the file
\texttt{intersect.lisp} in the supporting material.

In future work we plan to verify that our method for trapezoidal
restriction, as described in Section~\ref{sec:Sampling}, is correct.
Specifically we hope to show that restricted trapezoids are valid
generalizations and that they can be sampled to produce satisfying
variable assignments without backtracking, even for integer variables.
We also hope to extend our infrastructure to allow us to verify that
our approaches for generalizing uninterpreted function instances and
integer division and remainder, which we do not discuss in this
paper, are sound.

\section{Related Work}
\label{sec:Related}

We have compared the performance of our trapezoidal generalization
technique with the interval generalization capability provided by the
JKind model checker\cite{JKind}.  In all but the smallest models,
trapezoidal generalization was competitive with interval
generalization in terms of generalization speed.  For large models,
interval generalizations are actually more expensive to compute than
trapezoidal generalizations.  This is likely because trapezoidal
generalizations are computed in a single symbolic simulation of a
model whereas the interval generalization technique requires multiple
simulations to compute.  Within our fuzzing framework the sampling rates
for trapezoidal generalizations are also competitive with interval
generalization.  Tests in our framework suggest that interval
sampling rarely provides better than twice the overall performance of
trapezoidal sampling, and is typically less than 50 percent faster\footnote{Measurements
  are based on overall test vector generation rates, not
  generalization sampling rates in isolation}.
In terms of absolute performance, our deployed system routinely
exhibits effective trapezoidal sampling rates on the order of 100's of
thousands of generalized variables per second.

\begin{wrapfigure}{c}{0.2\textwidth}
  \centering
  \includegraphics[width=0.2\textwidth]{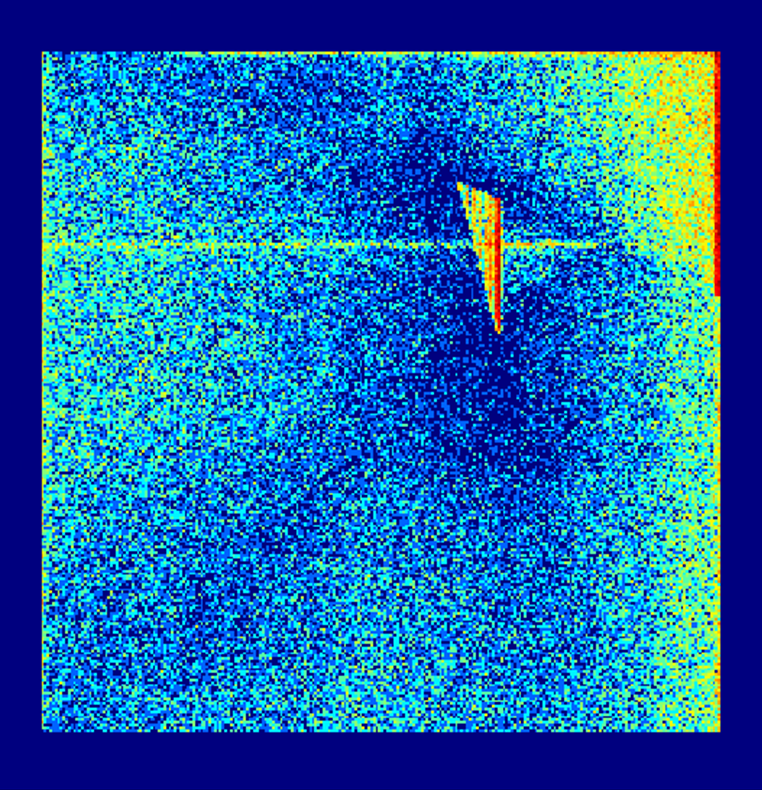}
  \includegraphics[width=0.2\textwidth]{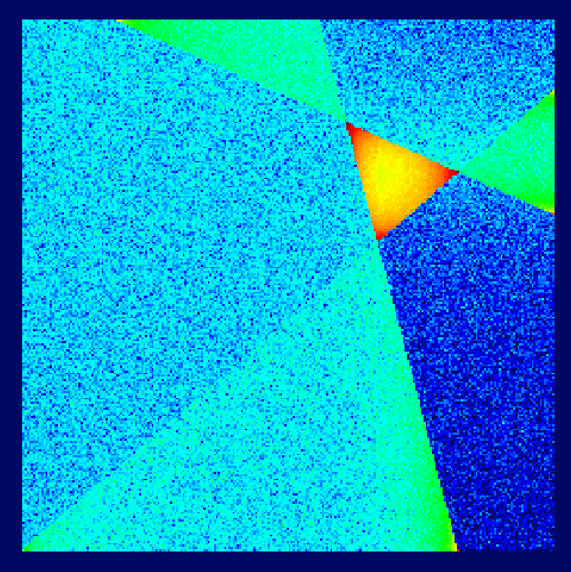}
\end{wrapfigure} 

Trapezoidal generalization also clearly outperforms interval
generalization in terms of preserving linear model features.  On the
right are two heat maps generated by capturing thousands of samples
from generalized solutions to a number of properties involving
arbitrary Boolean combinations of three linear constraints.  The upper
heat map employed interval generalization and the lower trapezoidal
generalization.  The three linear constraints involved in the various
properties are clearly visible in the lower heat map, whereas the upper
heat map is dominated by interval generalization artifacts.

Welp\cite{WelpQFBV} discusses a polytope generalization which is more
expressive than trapezoidal generalization but is also more expensive
to manipulate.  Operations on polytopes can be exponential in both
time and space.  The worst case size of our representation is
quadratic in the number of variables.  The intersection operation (on
ordered trapezoids) is worst case cubic and a naive implementation of
post-processing is worst case quartic for integer domains.  Sampling a
polytope is at least as hard as linear programming and for integers it
is NP-complete.  After post-processing, sampling a trapezoidal
generalization requires only a quadratic number of evaluations.
Trapezoidal generalization is thus more expressive than interval
generalization but more efficient than polytope generalization, both
from a computational point of view as well as from the perspective of
sampling.

Symbolic generalization techniques are also used in abstract
interpretation\cite{AbstractInterpretation}.  Because efficiency is an
overriding concern the abstractions commonly used therein are often
quite limited\cite{AbstractIntervals}.  While more complex and
expressive representations are typically less efficient, the octagon
domain\cite{DBLP:Octagon} is actually quite comparable to trapezoids
in terms of both computational and representational complexity.
Efficient sampling of such abstract domains, however, is generally not
a concern.  Also, the objective of generalization in abstract
interpretation is typically verification; generalization for this
purpose is common\cite{cegar}.  When used for verification,
generalization should provide a conservative over-approximations of
variable domains.  Our specification, on the other hand, calls for a
conservative under-approximation of such domains.  As a result, a
generalization technique designed for one domain would be unsound in
the other.

\section{Conclusion}

We have presented a generalization technique for logical formulae
described in terms of Boolean combinations of linear constraints that
is optimized for use with model-based fuzzing.  The generalization
technique employs trapezoidal solution sets that can be computed with
reasonable space and time bounds and then sampled efficiently.  We
developed a formal specification of generalization correctness and
used it to verify the correctness of key aspects of our generalization
algorithm using ACL2.  Finally we described a post-processing
procedure that produces restricted trapezoidal solutions that can be
rapidly and efficiently sampled without backtracking, even for integer
domains.  The formal verification of this post-processing step is
expected to be addressed in future work.

\bibliographystyle{eptcs}
\bibliography{ref}{}
\end{document}